\documentclass[a4paper]{article}

\usepackage{fullpage}

\usepackage[all=normal,bibliography=tight]{savetrees}

\usepackage{tikz}
\usepackage{xspace}
\usepackage{comment}

\usepackage{amsmath,amssymb,amsthm}
\usepackage{enumerate}
\usepackage{fancyvrb, url}
\RecustomVerbatimCommand{\VerbatimInput}{VerbatimInput}%
{fontsize=\scriptsize}

\usepackage{todonotes}

\newtheorem{theorem}{Theorem}
\newtheorem{lemma}[theorem]{Lemma}

\theoremstyle{definition}

\newtheorem{reduction}{Reduction Rule}
\newtheorem{reduction2}{Reduction Rule}
\newtheorem{branching}[reduction]{Branching Rule}
\newtheorem*{branching*}{Branching Rule}
\newtheorem*{subroutine*}{Subroutine}
\newtheorem{branching2}[reduction2]{Branching Rule}

\begin{document}

\newcommand{\Oh}{\ensuremath{\mathcal{O}}}
\newcommand{\Ohstar}{\ensuremath{\Oh^\ast}}
\newcommand{\fvsname}{\textsc{Feedback Vertex Set}}
\newcommand{\fvsshort}{\textsc{FVS}}
\newcommand{\disfvs}{\textsc{Disjoint-FVS}}
\newcommand{\zlota}{\phi}
\newcommand{\Vun}{U}
\newcommand{\Vdel}{D}
\newcommand{\pot}{\mu}
\newcommand{\sub}{\subseteq}
\newcommand{\sm}{\setminus}

\newcommand{\ccun}{\ell}
\newcommand{\ccdel}{\tau}
\newcommand{\tents}{t}

\newcommand{\defproblem}[4]{
  \vspace{2mm}
\noindent\fbox{
  \begin{minipage}{0.96\textwidth}
  \begin{tabular*}{0.96\textwidth}{@{\extracolsep{\fill}}lr} #1 & {\bf{Parameter:}} #3 \\ \end{tabular*}
  {\bf{Input:}} #2  \\
  {\bf{Question:}} #4
  \end{minipage}
  }
  \vspace{2mm}
}
\newcommand{\defnoparamproblem}[3]{
  \vspace{2mm}
\noindent\fbox{
  \begin{minipage}{0.96\textwidth}
  #1 \\
  {\bf{Input:}} #2  \\
  {\bf{Question:}} #3
  \end{minipage}
  }
  \vspace{2mm}
}

\title{Faster deterministic \textsc{Feedback Vertex Set}\thanks{Partially supported by NCN grant UMO-2012/05/D/ST6/03214 and Foundation for Polish Science.}}

\author{
  Tomasz Kociumaka\thanks{Institute of Informatics, University of Warsaw, Poland, \texttt{kociumaka@mimuw.edu.pl}}
  \and
  Marcin Pilipczuk\thanks{Institute of Informatics, University of Warsaw, Poland, \texttt{malcin@mimuw.edu.pl}}
}

\date{}

\maketitle

\begin{abstract}
We present two new deterministic  algorithms
for the \fvsname{} problem parameterized by the solution size. 
We begin with a simple algorithm, which runs in $\Ohstar((2+\zlota)^k)$ time,
where $\zlota<1.619$ is the golden ratio.  It already surpasses
the previously fastest $\Ohstar((1+2\sqrt{2})^k)$-time deterministic algorithm due to Cao et al.~[SWAT 2010].
In our developments we follow the approach of Cao et al., however, thanks to a
new reduction rule, we obtain not only better dependency on the parameter in the running
time, but also a solution with simple analysis and only a single branching rule.
Then, we present a modification of the algorithm which, using a more involved
set of branching rules, achieves $\Ohstar(3.592^k)$ running time.
\end{abstract}

\section{Introduction}

The \fvsname{} problem (\fvsshort{} for short), where we ask to delete
as few vertices as possible
from a given undirected graph to make it acyclic,
is one of the fundamental graph problems, appearing on the Karp's
list of 21 NP-hard problems \cite{karp}.
Little surprise it is also one of the most-studied problems in parameterized complexity,
and a long race for the fastest FPT algorithm (parameterized by the solution size, denoted $k$)
includes~\cite{fvs1,fvs2,fvs3,fvs5,fvs6,fvs7,guo:fvs,fvs:5k,fvs:3.83k,fvs:4krand,fvs:3k}.
Prior to this work, the fastest {\em{deterministic}} algorithm, due to Cao et al.~\cite{fvs:3.83k}, runs in $\Ohstar((1+2\sqrt{2})^k) \leq \Ohstar(3.83^k)$ time\footnote{The $\Ohstar$-notation suppresses factors polynomial in the input size.}; if we allow randomization,
the Cut\&Count technique yields an $\Ohstar(3^k$)-time algorithm~\cite{fvs:3k}.
Further research investigates
kernelization complexity of \fvsshort{} \cite{fvs:kernel1,fvs:kernel2,fvs:quadratic-kernel} and some generalizations e.g. to directed graphs~\cite{dfvs,sfvs,dsfvs}.

In this work we claim the lead in the `FPT race' for the fastest {\em{deterministic}}
algorithm for \fvsshort{}.

\begin{theorem}\label{thm:main-fast}
\fvsname{}, parameterized by the solution size $k$,
can be solved in $\Ohstar(3.592^k)$ time
and polynomial space.
\end{theorem}
\noindent
First, we present much simpler algorithm which proves a slightly weaker result.
\begin{theorem}\label{thm:main}
\fvsname{}, parameterized by the solution size $k$,
can be solved in $\Ohstar((2+\zlota)^k) \leq \Ohstar(3.619^k)$ time
and polynomial space where $\zlota = \frac{1 + \sqrt{5}}{2} < 1.619$ is the golden ratio.
\end{theorem}

In our developments, we closely follow the approach of 
the previously fastest algorithm due to Cao et al.~\cite{fvs:3.83k}.
That is, we first employ the iterative compression principle~\cite{reed:ic}
in a standard manner to reduce the problem to the disjoint compression variant (\disfvs{}),
where the vertex set is split into two parts, both inducing forests, and we are
allowed to delete vertices only from the second part.
Then we develop a set of reduction and branching rule(s) to cope with this
structuralized instance. We rely on the core observation of \cite{fvs:3.83k}
that the problem becomes polynomial-time solvable once the maximum degree of
the deletable vertices drops to $3$.

The main difference between our algorithm and the one of \cite{fvs:3.83k}
is the introduction of a new reduction rule
that reduces deletable vertices with exactly one deletable neighbour and two undeletable ones.
Branching on such vertices is the most costly operation
in the $\Ohstar(5^k)$ algorithm of Chen et al.~\cite{fvs:5k}
and avoiding such branching is a source of some
complications in the algorithm of \cite{fvs:3.83k}.
Introducing the new rule allows us to perform later only
a single straightforward branching rule.
Hence, the new rule not only leads to a better time complexity, but also
allows us to simplify the algorithm and analysis, comparing to \cite{fvs:3.83k}.

Additionally, we present a new, shorter proof of the main technical contribution
of \cite{fvs:3.83k}, asserting that \disfvs{} is polynomial-time solvable if
all deletable vertices are of degree at most $3$.
Thus, apart from better time complexity,
we contribute a simplification of the arguments of Cao et al.~\cite{fvs:3.83k}.

Finally, we modify our algorithm so that it achieves slightly better time
complexity. This requires
changes in the instance measure as well as introducing several branching rules.
The branching rules use each other as subroutines, which makes their branching
vectors long. Instead of manually determining this branching vectors, we provide
a script, which contains a compact and easily
readable representation of the branching rules and based on this data
determines the complexity of our algorithm.

\subsection{Preliminaries and notation}

All graphs in our work are undirected and, unless explicitly specified, simple.
For a graph $G$, by $V(G)$ and $E(G)$ we denote its vertex- and edge-set, respectively.
For $v \in V(G)$, the neighbourhood of $v$, denoted $N_G(v)$, is
defined as $N_G(v) = \{u \in V(G): uv \in E(G)\}$.
For a set $X \subseteq V(G)$, we denote the subgraph induced by $X$ as $G[X]$
and $G \setminus X$ denotes $G[V(G) \setminus X]$.
For $X\sub V(G)$ and $v\in V(G)$ we define an $X$-degree of $v$, denoted by
$\deg_X(v)$, as $|N_G(v) \cap X|$. Moreover, for $v\in X$ 
we say that $v$ is $X$-isolated if $\deg_X(v)=0$, and an $X$-leaf if $\deg_X(v)=1$.

If $uv$ is an edge in a (multi)graph $G$,
by {\em{contracting the edge $uv$}} we mean the following
operation: we replace $u$ and $v$ with a new vertex $x_{uv}$,
introduce $p-1$ loops at $x_{uv}$, where $p$ is the multiplicity of $uv$ in $G$,
and, for each $w \in (N_G(u) \cup N_G(v)) \setminus \{u,v\}$,
introduce an edge $wx_{uv}$ of multiplicity equal 
to the sum of the multiplicities of $wu$ and $wv$ in $G$.
In other words, we do not suppress multiple edges and loops in the process of contraction.
Note that, if $G$ is a simple graph and $N_G(u) \cap N_G(v) = \emptyset$,
no loop nor multiple edge is introduced when contracting $uv$.

\section{The simple algorithm}

\subsection{Iterative compression}

Following \cite{fvs:3.83k}, we employ the iterative compression principle~\cite{reed:ic}
in a standard manner.
Consider the following variant of \fvsshort{}.

\defnoparamproblem{\disfvs{}}{Graph $G$, a partition $V(G) = \Vun \cup \Vdel$
such that both $G[\Vun]$ and $G[\Vdel]$ are forests,
and an integer $k$.}{Does there exist a set $X \subseteq \Vdel$ of size at most $k$
such that $G \setminus X$ is a forest?}

A $\Vdel$-isolated vertex of degree $3$ is called a {\em{tent}}.
For a \disfvs{} instance $I=(G,\Vun,\Vdel,k)$ we define the following invariants:
$k(I) = k$,
$\ccun(I)$ is the number of connected components of $G[\Vun]$,
$\tents(I)$ is the number of tents in $I$,
and $\pot(I) = k(I) + \ccun(I) - \tents(I)$ is the {\em{measure}} of $I$.
Note that our measure differs from the one used in~\cite{fvs:3.83k}.
We omit the argument if the instance is clear from the context.

In the rest of the paper we focus on solving \disfvs{}, proving the following
theorem.
\begin{theorem}\label{thm:dismain}
\disfvs{} on an instance $I$
can be solved in $\Ohstar(\zlota^{\max(0,\pot(I))})$ time and polynomial space.
\end{theorem}

For sake of completeness, we show how Theorem \ref{thm:dismain}
implies Theorem \ref{thm:main}.

\begin{proof}[Proof of Theorem \ref{thm:main}]
Assume we are given a \fvsshort{} instance $(G,k)$.
Let $v_1,\ldots,v_n$ be an arbitrary ordering of $V(G)$.
Define $V_i = \{v_1,v_2,\ldots,v_i\}$, $G_i = G[V_i]$; we iteratively solve \fvsshort{}
instances $(G_i,k)$ for $i=1,2,\ldots,n$.
Clearly, if $(G_i,k)$ turns out to be a NO-instance for some $i$,
$(G,k)$ is a NO-instance as well.
On the other hand, $(G_i,k)$ is a trivial YES-instance for $i \leq k+1$.

To finish the proof we need to show how, given a solution $X_{i-1}$ to $(G_{i-1},k)$,
solve the instance $(G_i,k)$. Let $Z := X_{i-1} \cup \{v_i\}$ and
$\Vdel = V_i \setminus Z = V_{i-1} \setminus X_{i-1}$. Clearly,
$G[\Vdel]$ is a forest.
We branch into $2^{|Z|} \leq 2^{k+1}$ subcases, guessing the intersection
of the solution to $(G_i,k)$ with the set $Z$.
In a branch labeled $Y \subseteq Z$, we delete $Y$ from $G_i$ and disallow
deleting vertices of $Z \setminus Y$.
More formally, for any $Y \subseteq Z$ such that $G_i[Z \setminus Y]$ is a forest,
we define $\Vun = Z \setminus Y$ and apply the algorithm of Theorem \ref{thm:dismain}
to the \disfvs{} instance $I_Y = (G_i \setminus Y, \Vun, \Vdel, k-|Y|)$. 
Clearly, $I_Y$ is a YES-instance to \disfvs{} iff
$(G_i,k)$ has a solution $X_i$ with $X_i \cap Z = Y$.

As for the running time, note that $\ccun(I_Y) < |\Vun| = |Z \setminus Y| \leq (k+1)-|Y|$.
Hence, $\pot(I_Y) \leq 2(k-|Y|)$ and the total running time of solving
$(G_i,k)$ is bounded by
$$\Ohstar\left(\sum_{Y \subseteq Z} \zlota^{2(k-|Y|)}\right) = \Ohstar\left((1+\zlota^2)^k\right) = \Ohstar\left((2+\zlota)^k\right).$$
This finishes the proof of Theorem \ref{thm:main}.
\end{proof}

\subsection{Reduction rules}\label{ss:redPhi}

Assume we are given na \disfvs{} instance $I = (G,\Vun,\Vdel,k)$.
We first recall the (slightly modified) reduction rules of~\cite{fvs:3.83k}.
At any time, we apply the lowest-numbered applicable rule.

\begin{reduction}\label{red:01}
Remove all vertices of degree at most $1$ from $G$.
\end{reduction}

\begin{reduction}\label{red:two-nei}
If a vertex $v \in \Vdel$ has at least two neighbours in the same
connected component of $G[\Vun]$, delete $v$ and decrease $k$ by one.
\end{reduction}

\begin{reduction}\label{red:deg2}
If there exists a vertex $v \in \Vdel$ of degree $2$ in $G$,
move it to $\Vun$ if it has a neighbour in $\Vun$,
and contract one of its incident edges otherwise.
\end{reduction}

We shortly discuss the differences in the statements between our work
and~\cite{fvs:3.83k}. We say that a rule is {\em{safe}}
if the output instance is a YES-instance iff the input one is.
First, we apply Rule \ref{red:two-nei}
to any vertex with many neighbours in the same connected component of $G[\Vun]$,
not only to the degree-2 ones; however, the safeness of the new rule
is straightforward, as any solution to \disfvs{} on $I$ needs to contain such $v$.
Second, we prefer to contract an edge in Reduction \ref{red:deg2}
in case when $N_G(v) \subseteq \Vdel$, instead of moving $v$ to $\Vun$,
to avoid an increase of the measure $\pot(I)$. Note that,
as $G$ is simple and $G[\Vdel]$ is a forest, no multiple edge is 
introduced in such contraction and $G$ remains a simple graph.

By \cite{fvs:3.83k} and the argumentation above we infer that
Rules~\ref{red:01}--\ref{red:deg2} are safe
and applicable in polynomial time.
Let us now verify the following.
\begin{lemma}\label{lem:red-measure}
An application of any of the Rules \ref{red:01}--\ref{red:deg2}
does not increase $\pot(I)$.
\end{lemma}
\begin{proof}
Consider first an application of Rule \ref{red:01} to a vertex $v$.
If $v \in \Vdel$, $k(I)$, $\ccun(I)$ and $\tents(I)$ remains unchanged, as $v$
is not a tent.
If $v \in \Vun$, $k(I)$ remains unchanged and $\ccun(I)$ does not increase.
Note that $\tents(I)$ may decrease if the sole neighbour of $v$ is a tent. However,
in this case $\ccun(I)$ also drops by one, and $\pot(I)$ remains unchanged.

If Rule \ref{red:two-nei} is applied to $v$, $k(I)$ drops by one, $\ccun(I)$
remains unchanged and $\tents(I)$ remains the same or drops by one, depending
on whether $v$ is a tent or not.

If Rule \ref{red:deg2} is applied to $v$, $k(I)$ remains unchanged,
$\tents(I)$ does not decrease and $\ccun(I)$ does not increase
(thanks to the special case of $N_G(v) \subseteq \Vdel$).
\end{proof}

We now prove the lower bound on the measure $\pot(I)$ (cf. \cite{fvs:3.83k}, Lemma 1):
\begin{lemma}\label{lem:pot-bound}
Let $I$ be a \disfvs{} instance.
If $\tents(I) \geq k(I) + \frac{1}{2}\ccun(I)$, 
then $I$ is a NO-instance.
\end{lemma}
\begin{proof}
Assume $I = (G,\Vun,\Vdel,k)$ is a YES-instance and let $X$ be a solution.
Let $T \subseteq \Vdel$ be a set of tents in $I$.
For any $v \in T \setminus X$, $v$ connects three connected components
of $G[\Vun]$, hence $2|T \setminus X| < \ccun(I)$.
As $|X| \leq k$, $|T| < k + \frac{1}{2}\ccun(I)$ and the lemma follows.
\end{proof}
Consequently, we may apply the following rule.
\begin{reduction}\label{red:finish}
If $\pot(I) \leq \frac{1}{2}\ccun(I)$, conclude that $I$ is a NO-instance.
\end{reduction}
Note that Rule \ref{red:finish} triggers when $\pot(I) \leq 0$.

We now introduce a new reduction rule, promised in the introduction.
\begin{reduction}\label{red:new}
If $v$ is a $\Vdel$-leaf of $\Vun$-degree 2  with $w$ being its only
neighbour in $\Vdel$, subdivide the edge $vw$ and insert
the newly created vertex to $\Vun$.
\end{reduction}

First, we note that Rule \ref{red:new} is safe: introducing an undeletable
degree-2 vertex in a middle of an edge does not change the set of
feasible solutions to \disfvs{}. 
Second, note that Rule \ref{red:new} does not increase the measure of
the instance: although the newly added vertex creates a new connected
component in $G[\Vun]$, increasing $\ccun(I)$ by one, at the same time
$v$ becomes a tent and $\tents(I)$ increases by at least one
($w$ may become a tent as well). 
Third, one may be worried that
Rule \ref{red:new} in facts expands $G$, introducing a new vertex.
However,
it also increases the number of tents; as our algorithm never inserts a vertex
to $\Vdel$, Rule \ref{red:new} may be applied at most $|\Vdel|$ times.

\subsection{Polynomial-time solvable case}

Before we move to the branching rule, let us recall the polynomial-time
solvable case of Cao et al.~\cite{fvs:3.83k}.

\begin{theorem}[\cite{fvs:3.83k}]\label{thm:poly-case}
There exists a polynomial-time algorithm that
solves a special case of \disfvs{} where each vertex of $\Vdel$
is a tent.
\end{theorem}
Strictly speaking, in~\cite{fvs:3.83k} a seemingly more general case
is considered where each vertex of $\Vdel$ is of degree exactly three in $G$.
However, it is easy to see that an exhaustive application of Rule \ref{red:new}
to such an instance results in an instance where each vertex of $\Vdel$ is
a tent.

Theorem \ref{thm:poly-case} allows us to state the last reduction rule.
\begin{reduction}\label{red:poly-case}
If each vertex of $\Vdel$ is a tent, resolve the instance in polynomial time
using the algorithm of Theorem \ref{thm:poly-case}.
\end{reduction}

Theorem \ref{thm:poly-case} is proven in~\cite{fvs:3.83k}
by a reduction to the matroid parity problem in a cographic matroid.
We present here a shorter proof, relying on the matroid parity problem
in a graphic matroid of some contraction of $G$.

Given a (multi)graph $H$, the {\em{graphic matroid}} of $H$
is a matroid with ground set $E(H)$ where a set $A \subseteq E(H)$
is independent iff $A$ is acyclic in $H$ (spans a forest)%
\footnote{For basic definitions and results on matroids and graphic matroids, 
  we refer to the monograph~\cite{oxley}.}.

The {\em{matroid parity problem}} on the graphic matroid $H$
is defined as follows. In the input, apart from the (multi)graph $H$
with even number of edges, we are given a partition of $E(H)$
into pairs; in other words, $|E(H)| = 2m$ and
$E(H) = \{e_1^1, e_1^2, e_2^1,e_2^2, \ldots,e_m^1,e_m^2\}$.
We seek for a maximum set $J \subseteq \{1,2,\ldots,m\}$
such that $A(J) := \bigcup_{j \in J} \{e_j^1,e_j^2\}$ is independent,
i.e., is acyclic in $H$.
The matroid parity problem in graphic matroids is polynomial-time
solvable~\cite{DBLP:conf/icalp/GabowS85}.

Having introduced the matroid parity problem in graphic matroids, we
are ready to present a shorter proof of Theorem \ref{thm:poly-case}.

\begin{proof}[Proof of Theorem \ref{thm:poly-case}]
Let $I = (G,\Vdel,\Vun,k)$ be a \disfvs{} instance where
each vertex of $\Vdel$ is a tent.
For each $v \in \Vdel$ we arbitrarily enumerate the edges incident
to $v$ as $e_v^0, e_v^1, e_v^2$. 
Define $S = E(G[\Vun]) \cup \{e_v^0 : v \in \Vdel\}$
and let $H$ be the multigraph obtained from $G$ by contracting
all edges of $S$ (recall that we do not suppress the multiple edges and loops
in the process of contraction).
Clearly, $E(H) = \{e_v^1, e_v^2: v \in \Vdel\}$.
Treat $H$, with pairs $\{e_v^1,e_v^2\}_{v \in \Vdel}$ as an input
to the matroid parity problem in the graphic matroid of $H$.

Let $J \subseteq \Vdel$. We claim that $A(J) = \bigcup_{v \in J} \{e_v^1,e_v^2\}$
is independent in the graphic matroid in $H$ iff $G \setminus (\Vdel \setminus J)$
is a forest. Note this claim finishes the proof of Theorem \ref{thm:poly-case}
as it implies that a maximum solution $J$ to the matroid parity problem
corresponds to a minimum solution to \disfvs{} on $I$.

First note that, by the definition of $H$,
$A(J)$ is acyclic in $H$ iff $S \cup A(J)$ is acyclic in $G$.
Hence, if $A(J)$ is acyclic in $H$, $G \setminus (\Vdel \setminus J)$
is a forest, as $E(G \setminus (\Vdel \setminus J)) \subseteq S \cup A(J)$.
In the other direction, assume that $G \setminus (\Vdel \setminus J)$ is acyclic.
Note that
$$S \cup A(J) = E(G \setminus (\Vdel \setminus J)) \uplus \{e_v^0: v \in \Vdel \setminus J\}.$$
However, each vertex $v \in \Vdel \setminus J$ has only one incident edge that belongs to $S \cup A(J)$. Hence, $S \cup A(J)$ is acyclic
and the claim is proven.
\end{proof}

\subsection{Branching rule}\label{ss:branch}

We conclude with a final branching rule.

\begin{branching}\label{branching}
Pick a vertex $v \in \Vdel$ that is not a tent and has a maximum possible number
of neighbours in $\Vun$. Branch on $v$: either delete $v$ and decrease $k$ by one
or move $v$ to $\Vun$.
\end{branching}

First, note that the branching of Rule \ref{branching} is exhaustive:
in first branch, we consider the case when $v$ is included in the
solution to the instance $I$ we seek for, and in the second branch --- when the
solution does not contain $v$.

In the branch where $v$ is removed, the number of tents does not decrease
and the number of connected components of $G[\Vun]$ remains the same.
Hence, the measure $\pot(I)$ drops by at least one.

Let us now consider the second branch.
As Rule \ref{red:poly-case} is not applicable, there
exists a connected component of $G[\Vdel]$ that is not a tent,
and there exists a vertex $u$ in this component whose degree in $G[\Vdel]$
is at most one. As Rules \ref{red:01}, \ref{red:deg2} and \ref{red:new}
are not applicable, $u$ has at least three neighbours in $\Vun$.
As $u$ is not a tent, Rule \ref{branching} may choose $u$
and, consequently, it triggers on a vertex $v$ with
at least $3$ neighbours in $\Vun$.
As Rule \ref{red:two-nei} is not applicable, in the branch
where $v$ is moved to $\Vun$, $G[\Vun]$ remains a forest
and the number of its connected components, $\ccun(I)$, drops by at least two. 
Hence, as $\tents(I)$ does not decrease and $k(I)$ remains unchanged in this branch,
the measure $\pot(I)$ drops by at least two.

By Lemma \ref{lem:red-measure}, no reduction rule may increase
the measure $\pot(I)$. Rule \ref{red:finish} terminates computation if $\pot(I)$
is not positive.
The branching rule, Rule \ref{branching}, yields drops of measure $1$ and $2$ in
the two considered subcases. We conclude that the algorithm
resolves an instance $I$ in $\Ohstar(\zlota^{\pot(I)})$ time and polynomial space,
concluding the proof of Theorem \ref{thm:dismain}.

\section{The $\Ohstar(3.592^k)$-time algorithm}
In the improved algorithm we still follow the iterative compression principle
and work with \disfvs{} problem. 

Let $I=(G,D,U,k)$ be an instance of \disfvs{}. For technical reasons we 
no longer require that $G[U]$ is forest, despite the fact that whenever
$G[U]$ contains a cycle, $I$ is clearly a NO-instance. Let
us define $\ell'(I) = |U|-|E(G[U])|$. Note that $\ell'(I)$ generalizes $\ell(I)$ 
defined for instance with acyclic $G[U]$ as $\ell'(I)$ equals then the number of connected components
of $G[U]$.

For a real constant $\alpha\in \left[\tfrac{1}{2},1\right]$ we define
$\mu_\alpha(I) = k(I)+\alpha \ell'(I)-t(I)$. Recall that the previous algorithm used $\alpha=1$
while in \cite{fvs:3.83k} this parameter is set to $\frac{1}{2}$.

\subsection{Reduction rules}
We use reduction rules similar to those defined Section~\ref{ss:redPhi}.
Since we use a more general measure and some of the rules differ, we explicitly
state all the rules we use and briefly discuss their correctness.

\begin{reduction2}\label{r:zero}
Remove all vertices $v\in \Vdel$ of degree at most 1.
\end{reduction2}
\begin{reduction2}\label{r:one}
If there exists a vertex $v\in \Vdel$ of degree 2, move it to $U$ if it has a neighbour in $U$, and
contract one of its incident edges otherwise.
\end{reduction2}
\begin{reduction2}\label{r:meas}
If $\mu_\alpha(I)\le (\alpha-\frac{1}{2})\ell'(I)$, conclude that $I$ is a
NO-instance.
\end{reduction2}
\begin{reduction2}\label{r:two}
If $v$ is a $\Vdel$-leaf of $\Vun$-degree 2  with $w$ being its only
neighbour in $\Vdel$, subdivide the edge $vw$ and insert
the newly created vertex to $\Vun$.
\end{reduction2}
\begin{reduction2}\label{r:tent}
If each $v\in \Vdel$ is a tent, resolve the instance in polynomial time using
the algorithm of Theorem~\ref{thm:poly-case} if $G[U]$ is a forest, and return NO otherwise.
\end{reduction2}

\begin{lemma}
Rules~\ref{r:zero}--\ref{r:tent} are safe. Moreover, each of them either
immediately gives the answer or decreases $|\Vdel|$ simultaneously not
increasing the measure $\mu_\alpha$
\end{lemma}
\begin{proof}
All rules except Rule~\ref{r:meas} are safe since their counterparts in
Section~\ref{ss:redPhi} were safe.
To prove the safeness of Rule~\ref{r:meas} observe that if $G[U]$ is a forest,
then the instance is trivially a NO-instance and otherwise $\ell(I)$ coincides
with $\ell'(I)$ so Lemma~\ref{lem:pot-bound} immediately shows that $I$ is a
NO-instance if the rule is applicable.

Now, let us analyze the change in measure.
Rule~\ref{r:zero} leaves $k$ and $\ell'$ unchanged while $t$
may only increase since the neighbour of the vertex removed from $\Vdel$ might
become a tent. 
Rule~\ref{r:one} leaves $k$ unchanged, and depending on the case, either 
does not modify $G[\Vun]$ or introduces one vertex and one or two edges to
$G[U]$.
Consequently, $\ell'$ does not increase. The rule may create a tent,
but does not remove any, so $t$ does not decrease.
Rule~\ref{r:two} introduces one vertex to $G[\Vun]$, so it
increases $\ell'$ by one. Simultaneously, it leaves $k$ unchanged and introduces
at least one tent. In total, this gives a drop of at least $(1-\alpha)$ in the
measure.
\end{proof}

We apply the reduction rules in any order until we obtain an \emph{irreducible}
instance, i.e. no rule is applicable.
Let us gather some properties of such instances.

\begin{lemma}\label{lem:prop}
Assume that $I=(G,\Vun, \Vdel, k)$ is an irreducible instance. Then
\begin{enumerate}[(a)]
    \item\label{pc} each $v\in \Vdel$ of $\Vun$-degree 0 has $\Vdel$-degree at
  least 3,
  \item\label{pa} $\mu_\alpha(I) > 0$,
  \item\label{pd} there is at least one $v\in \Vdel$ which is not a tent,
  \item\label{pb} each $\Vdel$-leaf and $\Vdel$-isolated vertex has
  $\Vun$-degree 3 or more.
\end{enumerate}
\end{lemma}
\begin{proof}
Property~(\ref{pc}) is an immediate consequence of
inapplicability of Rules~\ref{r:zero} and \ref{r:one}, property~(\ref{pa}) of Rule~\ref{r:meas}
and property~(\ref{pd}) of Rule~\ref{r:tent}. For a proof of~(\ref{pb})
observe that $\Vdel$-isolated vertices of degree 0 and 1 are eliminated by
Rule~\ref{r:zero}, while of degree 2 by Rule~\ref{r:one}. Similarly,
$\Vdel$-leaves of $\Vun$ degree 0,1 and 2 are dismissed by Rules~\ref{r:zero},
\ref{r:one} and~\ref{r:two} respectively.
\end{proof}

\subsection{Branching rules}
If no reduction rules is applicable, the algorithm performs one of the branching
rules. The branching rules are built from several elementary
operations. These operations include performing Rules~\ref{r:zero},
\ref{r:one} or~\ref{r:two} as well as the following branching rule:
\begin{branching2}[Branching on $v$]\label{b:elem}
Let $v\in \Vdel$. Either delete $v$
and decrease $k$ by one (\emph{delete} branch) or move $v$ to $\Vun$ (\emph{fix}
branch).
\end{branching2}
Similarly to
the branching rule of the previous algorithm, this rule is clearly exhaustive.
Observe that if $v$ is not a tent, then in the \emph{delete} branch the measure
decreases by at least one (since $k$ decreases) and in the \emph{fix} branch, the
measure decreases by at least $\alpha(f-1)$, where $f$ is the $\Vun$-degree of
$v$.
This is because one vertex and $f$ edges are moved to $G[\Vun]$. 
Also, if $v$ has a parent, its $\Vun$-degree remains unchanged in the
\emph{delete} branch and raises by one in the \emph{fix} branch.

With all building blocks ready, let us proceed with the description of the
structure steering execution of the algorithm.
For each connected component $T$ of $V[\Vdel]$ we select one of the vertices
of $T$ as the \emph{root} of $T$. We require the root to be either $\Vdel$-isolated
or a $\Vdel$-leaf. The roots are selected every time we need to choose a branching rule,
each such selection can be performed independently.
With a root in each component, we can define the parent-child relation on
$\Vdel$, which we use to partition vertices of $\Vdel$ into four types:

\begin{enumerate}[(a)]
  \item \emph{tents}
  \item other vertices of $\Vun$-degree 3 with no children, called \emph{singles},
  \item vertices of $\Vun$-degree 0 with two children, both singles, called \emph{doubles},
  \item the remaining vertices, called \emph{standard} vertices.
\end{enumerate}
 
We call a standard vertex a \emph{guide} if none of its children is standard.
Each guide has three parameters: the $\Vun$-degree $f$, the number of children being singles $s$, and
the number of children being doubles $d$. We call the triple $(f,s,d)$ 
the \emph{type} of a guide. In all branching rules we pick an arbitrary guide
$v$ and perform operations on the vertices in a subtree rooted on $v$. The choice of
the branching rule depends on the type of the guide.
The following lemma justifies such an approach.
\begin{lemma}
If $I=(G,\Vun, \Vdel, k)$ is irreducible, $\Vdel$ contains a guide.
\end{lemma}
\begin{proof}
Clearly, it suffices to prove that there exists a standard vertex in $\Vdel$.
Observe that if $v$ is standard, so is its parent (if any). 
Consequently, we shall prove that a root of some component of $G[\Vdel]$ is a
standard vertex.
By Lemma~\ref{lem:prop}(\ref{pd}) there is a vertex $v\in \Vdel$, which is not a
tent. Let $T$ be its connected component of $G[\Vdel]$ and $r$ be a root of $T$.
 Note that, since roots were chosen to have $\Vdel$-degree at most 1, by
 Lemma~\ref{lem:prop}(\ref{pd}) $r$ has $\Vun$-degree at least 3. 
 Thus, $r$ is not a double. If $r$ were a single, it would be $\Vdel$-isolated,
 then, however, it would either be a tent or its
 $\Vun$-degree would be at least 4. The former case is impossible by the choice
 of $T$ while in the latter $r$ is standard as desired.
 Therefore $r$ is standard.
\end{proof}
 
Finally, observe that some triples of non-negative integers cannot be types of a
guide. Indeed, types $(0,0,1)$ and $(0,1,0)$ contradict
Lemma~\ref{lem:prop}(\ref{pc}), while types $(0,0,0)$, $(1,0,0)$ and $(2,0,0)$
Lemma~\ref{lem:prop}(\ref{pb}).
Moreover, types $(3,0,0)$ and $(0,2,0)$ are forbidden, since a guide of this
type would actually be a single, tent or a double.
Also note that, since any root has $\Vun$-degree at least 3, a guide with $f\le 2$ always has a parent.

Let us start with a pair of rules, which are used as subroutines in other rules. All branching rules are indexed by guide types. 
We use $\ge n$ as `at least $n$'. If a guide matches several rules, any of them
can be applied. In all the descriptions the guide is called $v$. 
\begin{branching*}[$\ge\!4,0,0$] Branch on $v$ (i.e. perform Rule \ref{b:elem}).
\end{branching*}
\begin{branching*}[$1, 1, 0$] Branch on $w$, the only child of $v$.
In the \emph{delete} branch, use Rule~\ref{r:one} to move
$v$ to $\Vun$. In the \emph{fix} branch, use Rule~\ref{r:two} to make $v$ a
tent. In both branches, the $\Vun$-degree of a parent of $v$ raises by one.
\end{branching*}

Before we list the remaining rules, let us describe a subroutine which is
not used as a branching rule on its own.

\begin{subroutine*}[Eliminate a double]
Let $v$ be a guide of type $(f,s,d)$ and let its child $u$ be a double with
children $w_1$, $w_2$. 
Branch on $w_1$. In the \emph{delete} branch, use Rule~\ref{r:one} to dissolve
$u$.
Then, $v$ becomes an $(f,s+1,d-1)$-guide.
In the \emph{fix} branch, $u$ becomes a $(1,1,0)$-guide, proceed with the
$(1,1,0)$ rule.  In both branches of this rule, $v$ becomes an
$(f+1,s,d-1)$-guide.
\end{subroutine*}

Below, we give branching rules for all types of guides not considered yet.
Most of these rules use the same scheme, but for sake of clarity and
consistence with complexity analysis, we state all rules explicitly.

\begin{branching*}[$\ge2,\ge1,\ge0$] Let a single $w$ be a child of
$v$. Branch on $v$. In the \emph{delete} branch $w$ becomes a tent, in the \emph{fix}
branch proceed as in the $(\ge 4, 0, 0)$ rule for $w$.
\end{branching*}
\begin{branching*}[$\ge2,\ge 0,\ge1$] Let a double $w$ be a child of $v$.
Use the subroutine to eliminate $w$. In the \emph{delete} branch proceed as in
the $(\ge 2, \ge 1, \ge 0)$ rule. In the remaining branches,
do not do anything more.
\end{branching*}
\begin{branching*}[$1, 0, 1$] Let a double $w$ be the only child of $v$.
Use the subroutine to eliminate $w$. In the \emph{delete} branch
proceed as in the $(1,1,0)$-rule, in the remaining branches observe that $v$
becomes a $\Vdel$-leaf of $\Vun$-degree 2, and apply Rule~\ref{r:two} to
eliminate it.
\end{branching*}
\begin{branching*}[$1, \ge 2, \ge 0$]
Let singles $w_1,w_2$ be children of $v$. Branch on $v$.
In the \emph{delete} branch $w_1$ and $w_2$ become tents. In the `remove' branch,
eliminate both $w_1$ and $w_2$ with the $(\ge 4, 0,0)$ rule.
\end{branching*}
\begin{branching*}[$1, \ge 1, \ge 1$]
Let a double $w$ be a child of $v$. Use the subroutine to eliminate $w$. In the
\emph{delete} branch proceed as in the $(1, \ge 2, \ge 0)$ rule, in the remaining
branches proceed as in the $(\ge 2, \ge 1, \ge 0)$ rule.
\end{branching*}
\begin{branching*}[$1, \ge 0, \ge 2$]
Let a double $w$ be a child of $v$. Use the subroutine to eliminate $w$. In the
\emph{delete} branch proceed as in the $(1, \ge 1, \ge 1)$ rule, in the remaining
branches proceed as in the $(\ge 2, \ge 0, \ge 1)$ rule.
\end{branching*}
\begin{branching*}[$0, 1, 1$]
Let a double $w$ be a child of $v$. Use the subroutine to eliminate $w$. In the
\emph{delete} branch do not do anything more, in the remaining branches proceed as
in the $(1,1,0)$ rule.
\end{branching*}
\begin{branching*}[$0, 0, 2$]
Let a double $w$ be a child of $v$. Use the subroutine to eliminate $w$. In the
\emph{delete} branch proceed as in the $(0,1,1)$ rule, in the remaining branches
proceed as in the $(1,0,1)$ rule.
\end{branching*}
\begin{branching*}[$0, \ge 3, \ge 0$]
Let a single $w$ be a child of $v$. Branch on $w$. In the \emph{delete} branch
do not do anything more, in the \emph{fix} branch proceed as in $(1, \ge 2, \ge
0)$ rule.
\end{branching*}
\begin{branching*}[$0, \ge 2, \ge 1$]
Let a double $w$ be a child of $v$. Use the subroutine to eliminate $w$. In the
\emph{delete} branch proceed as in the $(0, \ge 3, \ge 0)$ rule, in the remaining
branches proceed as in the $(1, \ge 2, \ge 0)$ rule.
\end{branching*}
\begin{branching*}[$0, \ge 1, \ge 2$]
Let a double $w$ be a child of $v$. Use the subroutine to eliminate $w$. In the
\emph{delete} branch proceed as in the $(0, \ge 2, \ge 1)$ rule, in the remaining
branches proceed as in the $(1, \ge 1, \ge 1)$ rule.
\end{branching*}
\begin{branching*}[$0, \ge 0, \ge 3$]
Let a double $w$ be a child of $v$. Use the subroutine to eliminate $w$. In the
\emph{delete} branch proceed as in the $(0, \ge 1, \ge 2)$ rule, in the remaining
branches proceed as in the $(1, \ge 0, \ge 2)$ rule.
\end{branching*}

It is easy to see that for each guide some reduction rule is applicable.
Moreover, the branch tree for each rules is of constant size. The rules are
built of elementary operations which decrease $\Vdel$, so we obtain
a correct algorithm solving the \disfvs{} problem.

\subsection{Complexity analysis}
Due to numerous branching rules, some of them rather complicated, a manual
complexity analysis would be tedious and error-prone.
Therefore in the Appendix we provide a Python script, which automates the
analysis.

Our script computes the \emph{branching vectors}\footnote{We refer to
~\cite{fomin2010exact} for more on branching vectors.}.
It relies on a description of all branching rules, strictly following
their definitions presented above. Computing the branching vectors, 
we use the following facts, all proved above. 
\begin{enumerate}[(a)]
  \item Rules~\ref{r:zero} and~\ref{r:one} do not increase the measure.
  \item Rule~\ref{r:two} decreases the measure by at least $1-\alpha$.
  \item Obtaining a tent corresponds to drop in measure equal to $1$.
  \item Elementary branch on a vertex $v$ with $\Vun$-degree $f$ has measure
  drop at least $1$ in the \emph{delete} branch and $(f-1)\alpha$ in the \emph{fix}
  branch.
\end{enumerate}

Having computed the branching vectors, our script determines
$\beta_\alpha$: the maximum positive root of the corresponding equations over
all vectors.
Standard reasoning for branching algorithms lets us conclude that our algorithm
for  \disfvs{} works in $\Ohstar(\beta_\alpha^{\mu_\alpha(I)})$ time.

The instances of \disfvs{} arising from iterative compression of \fvsshort{} have
$|\Vun|=k+1$, so $\ell' \le k+1$.
Thus, for such an instance $I$ we have $\mu_\alpha(I)=
k(I)+\alpha\ell'(I)-t(I)\le k(1+\alpha)+1$.
Consequently, our algorithm solves them in
$\Ohstar\left(\left(\beta_\alpha\right)^{(1+\alpha)k}\right)$ time and
polynomial space.
For $\alpha=0.84$ this becomes $\Ohstar(2.592^k)$.
Repeating the reasoning from the proof of Theorem~\ref{thm:main}, we 
complete the proof of Theorem~\ref{thm:main-fast}.


\section{Concluding remarks}

In our paper we presented a two deterministic FPT algorithms for \fvsname{}.
The former can be seen as a reinterpretation and simplification of the previously
fastest algorithm of Cao et al.~\cite{fvs:3.83k}. The latter algorithm
performs branches more carefully and consequently its running time is slightly
better. 

However, we are still far from matching the time complexity of the best
{\em{randomized}} algorithm, using the Cut\&Count technique \cite{fvs:3k}.
In particular, in our work we do not use any insights both from the
Cut\&Count technique and its later derandomizations \cite{cc:derand1,cc:derand2}.
Obtaining a {\em{deterministic}} algorithm for \fvsshort{}
running in $\Ohstar(3^k)$ time remains a challenging open problem.

We would also like to note that we are not aware of any lower bounds
for FPT algorithms of \fvsshort{} that are stronger than
a refutation of a subexponential algorithm, based on the Exponential
Time Hypothesis, that follows directly from a similar result for \textsc{Vertex Cover}%
\footnote{Replace each edge with a short cycle to obtain a reduction from \textsc{Vertex Cover}
to \fvsname{}.}.
Can we show some limits for FPT algorithms for \fvsshort{}, assuming
the Strong Exponential Time Hypothesis, as it was done for e.g.
\textsc{Steiner Tree} or \textsc{Connected Vertex Cover}~\cite{ccc}?

\bibliographystyle{abbrv}
\bibliography{fvs-golden}

\newpage
\appendix
\section*{Python script automating complexity analysis\footnote{
Also available at~\url{students.mimuw.edu.pl/~kociumaka/fvs}}}

\begin{Verbatim}
import scipy.optimize

def value(vector):
    """compute the value of a branching vector"""
    def h(x):
        return sum([x**(-v) for v in vector])-1
    return scipy.optimize.brenth(h,1, 100)


def join(first, then):
     """peform 'then' in each branch after the execution of 'first' """
     return [x+y for x in first for y in then]


alpha = 0.84

no   = [0]        # trivial continuation, do not do anything more
zero = [0]        # perform reduction rule 1
one  = [0]        # perform reduction rule 2
two  = [1-alpha]  # perform reduction rule 4
tent = [1]        # discover a tent

def br (f, delete, fix): 
    """perform an  elementary branch on a vertex of U-degree f"""
    return join([1], delete) + join([alpha*(f-1)], fix)
    
guide = dict()
guide[('1','1','0')] = br(3, one, two)
guide[('>=4','0','0')] = br(4, no, no)

def double(delete, fix):
    """perform the double elimination subroutine"""
    return br(3, delete, br(3, join(one, fix), join(two, fix)))

guide[('>=2','>=1','>=0')] = br(2, tent, guide[('>=4','0','0')])
guide[('>=2','>=0','>=1')] = double(guide[('>=2','>=1','>=0')], no)
guide[('1','0','1')]       = double(guide[('1','1','0')], two)
guide[('1','>=2','>=0')]   = br(1, join(tent, tent), join(guide[('>=4','0','0')],guide[('>=4','0','0')]))
guide[('1','>=1','>=1')]   = double(guide[('1','>=2','>=0')], guide[('>=2','>=1','>=0')])
guide[('1','>=0','>=2')]   = double(guide[('1','>=1','>=1')], guide[('>=2','>=0','>=1')])
guide[('0','1','1')]       = double(no, guide[('1','1','0')])
guide[('0','0','2')]       = double(guide[('0','1','1')] , guide[('1','0','1')])
guide[('0','>=3','>=0')]   = br(3,no, guide[('1','>=2','>=0')])
guide[('0','>=2','>=1')]   = double(guide[('0','>=3','>=0')], guide[('1','>=2','>=0')])
guide[('0','>=1','>=2')]   = double(guide[('0','>=2','>=1')], guide[('1','>=1','>=1')])
guide[('0','>=0','>=3')]   = double(guide[('0','>=1','>=2')], guide[('1','>=0','>=2')])
    

vectors = guide.values()

print max([value(b) for b in vectors])**(1+alpha)
\end{Verbatim}

\end{document}